\documentclass[11pt,english,conference]{IEEEtran}
\usepackage[T1]{fontenc}
\usepackage[latin9]{inputenc}
\usepackage{units}
\usepackage{amsthm}
\usepackage{amsmath}
\usepackage{graphicx}
\usepackage{amssymb}

\makeatletter

\providecommand{\tabularnewline}{\\}

  \theoremstyle{plain}
  \newtheorem{lem}{Lemma}
 \ifx\proof\undefined\
   \newenvironment{proof}[1][\proofname]{\par
     \normalfont\topsep6\p@\@plus6\p@\relax
     \trivlist
     \itemindent\parindent
     \item[\hskip\labelsep
           \scshape
       #1]\ignorespaces
   }{%
     \endtrivlist\@endpefalse
   }
   \providecommand{\proofname}{Proof}
 \fi
  \theoremstyle{plain}
  \newtheorem{thm}{Theorem}
  \theoremstyle{remark}
  \newtheorem{rem}{Remark}
  \theoremstyle{plain}
  \newtheorem{cor}{Corollary}

\makeatother

\usepackage{babel}

\begin{document}

\title{Subset Typicality Lemmas and Improved Achievable Regions in Multiterminal
Source Coding }

\author{Kumar Viswanatha, Emrah Akyol and Kenneth Rose\\
ECE Department, University of California - Santa Barbara\\
\{kumar,eakyol,rose\}@ece.ucsb.edu
\thanks{The work was supported by the NSF under grants CCF-0728986, CCF - 1016861 and  CCF-1118075} }
\maketitle
\begin{abstract}
Consider the following information theoretic setup wherein independent
codebooks of $N$ correlated random variables are generated according
to their respective marginals. The problem of determining the conditions
on the rates of codebooks to ensure the existence of at least one
codeword tuple which is jointly typical with respect to a given joint
density (called the multivariate covering lemma) has been studied
fairly well and the associated rate regions have found applications
in several source coding scenarios. However, several multiterminal
source coding applications, such as the general multi-user Gray-Wyner
network, require joint typicality \textit{only} within subsets of
codewords transmitted. Motivated by such applications, we ask ourselves
the conditions on the rates to ensure the existence of at least one
codeword tuple which is jointly typical within subsets according to
given per subset joint densities. This report focuses primarily on
deriving a new achievable rate region for this problem which strictly
improves upon the direct extension of the multivariate covering lemma,
which has quite popularly been used in several earlier work. Towards
proving this result, we derive two important results called `subset
typicality lemmas' which can potentially have broader applicability
in more general scenarios beyond what is considered in this report.
We finally apply the results therein to derive a new achievable region
for the general multi-user Gray-Wyner network.\end{abstract}

\section{Introduction\label{sec:Introduction-1}}

Consider a scenario where independent codebooks of $N$ random variables
$(X_{1},X_{2},\ldots,X_{N})$ are generated according to some given
marginal distributions at rates $(R_{1},R_{2},\ldots,R_{N})$ respectively.
Let $\mathcal{S}_{1},\mathcal{S}_{2}\ldots\mathcal{S}_{M}$ be $M$
subsets of $\{1,2,\ldots,N\}$ and let the joint distributions of
$(X_{1},X_{2},\ldots,X_{N})$ within each subset, consistent with
each other and with the marginal distributions, be given. We ask ourselves
the conditions on the rates $(R_{1},R_{2},\ldots,R_{N})$ (achievable
region) so that the probability of finding one codeword from each
codebook, such that the codewords are all jointly typical within subsets
$\mathcal{S}_{1},\mathcal{S}_{2}\ldots\mathcal{S}_{M}$ according
to the given per subset joint distributions, approaches $1$. We denote
the given probability distribution over subset $\mathcal{S}_{i}$
by $P(\{X\}_{\mathcal{S}_{i}})$. The conditions on the rates when
$\mathcal{S}_{i}=\{1,\ldots,N\}$, i.e, when the joint distribution
over all the random variables is given, can be derived using standard
typicality arguments and is quite popularly called as the multivariate
covering lemma \cite{EG_LN,VKG}%
\footnote{We note that the underlying principles and proofs of multivariate
covering lemma appeared much earlier in the literature, for example
\cite{EGC}. However the nomenclature and the general applicability
of the underlying ideas have been elucidated quite clearly in \cite{EG_LN}%
}. It says that for any joint density over $(X_{1},X_{2},\ldots,X_{N})$,
if the codebooks are generated according to the respective marginals,
the probability of not finding a jointly typical codeword tuple approach
$0$ if $\forall\mathcal{J}\subseteq\{1,2,\ldots,N\}$:\begin{equation}
\sum_{i\in\mathcal{J}}R_{i}\geq\sum_{i\in\mathcal{J}}H(X_{i})-H(P(\{X\}_{\mathcal{J}}))\label{eq:MCL}\end{equation}
where $\{X\}_{\mathcal{J}}$ denotes the set $X_{i}:i\in\mathcal{J}$
and $H(P)$ denotes the entropy of any distribution $P$. 

A fairly direct extension of the multivariate covering lemma, to the
more general scenario of arbitrary subsets $\mathcal{S}_{1},\mathcal{S}_{2}\ldots\mathcal{S}_{M}$,
which has been quite popularly used in several information theoretic
scenarios, such as \cite{VKG,Ramchandran,MD_ISIT,fusion,DIR}, can
be described as follows. Fix any joint density $\tilde{P}(X_{1},X_{2},\ldots,X_{N})$
such that:\begin{equation}
\tilde{P}(\{X\}_{\mathcal{S}_{j}})=P_{\mathcal{S}_{j}}(\{X\}_{\mathcal{S}_{j}})\,\,\forall j\label{eq:MCL2}\end{equation}
i.e, it satisfies the given joint distributions within subsets $\mathcal{S}_{j}\,\,\forall j$.
Then the set of all rate tuples satisfying the following conditions
are achievable, $\forall\mathcal{J}\subseteq\{1,2,\ldots,N\}$:\begin{equation}
\sum_{i\in\mathcal{J}}R_{i}\geq\sum_{i\in\mathcal{J}}H(X_{i})-H(\tilde{P}(\{X\}_{\mathcal{J}}))\label{eq:MCL3}\end{equation}
The convex closure of all achievable rate tuples, over all such joint
densities $\tilde{P}$ satisfying the given per subset densities is
an achievable region for the problem. We denote this region by $\mathcal{R}_{a}$.
Our primary objective in this report is to show that the rate region
in (\ref{eq:MCL3}) with the individual functionals set to their respective
maxima subject only to their specific exact constraints is, infact,
achievable. Specifically we show that, each of the terms $H(\tilde{P}(\{X\}_{\mathcal{J}}))$
can be replaced with the corresponding maximum entropy functionals
$H^{*}(\tilde{P}(\{X\}_{\mathcal{J}}))$ subject to only the constraints
pertinent to subsets of $\{X\}_{\mathcal{J}}$. This allows us to
achieve simultaneous optimum of all the functionals leading to a strictly
larger achievable region than $R_{a}$. Towards proving this result,
we establish two important lemmas, namely `subset typicality lemmas',
which may prove to have much wider applicability in general scenarios
beyond the scope of this report. 

Scenarios depicted in the above example, where typicality within subsets
of codewords is sufficient for decoding, arise quite frequently in
several multiterminal source coding setups. One of the most typical
examples is the multi-user generalization of the Gray-Wyner network
\cite{GW} discussed in section \ref{sec:GW} where the encoder observes
$K$ random variables and there are $K$ sinks, each decoding one
of the random variables upto a prescribed distortion constraint%
\footnote{We note that \cite{multi_user_privacy} considers a particular generalization
of the Gray-Wyner network to multiple users with applications in information
theoretic security where a unique common branch is sent to all the
decoders along with their respective individual rates. However we
assert that the most general extension of the 2 user Gray-Wyner network
will involve a combinatorial number of branches, each being sent to
a unique subset of the decoders.%
}. The most general setting involves $2^{K}-1$ branches (encoding
rates), each being sent to a unique subset of the decoders. Observe
that it is sufficient if all the codewords being sent to sink $i$
are jointly typical with the $i$'th source sequence and enforcing
joint typicality of all the codewords in an unnecessary restriction.
Similar settings arise in the context of dispersive information routing
of correlated sources \cite{DIR}, fusion coding and selective retrieval
in a database \cite{fusion} and in several other scenarios which
can be considered as particular cross-sections of the general $L-$channel
`multiple descriptions' (MD) problem \cite{VKG,Ramchandran}. We note
that, in this report, we demonstrate the workings of the underlying
principle in the context of the example we described above. However
it is important to note that the results we derive have implications
in a wide variety of problems involving optimization of multiple functionals,
each depending on a subset of the random variables, subject to constraints
on their joint distributions.

\section{Main Results \label{sec:Main_result}}

In this section, we first establish the subset typicality lemmas which
will finally lead to Theorems \ref{thm:main_tm} and \ref{lem:eg}
showing strictly larger achievable rates compared to $\mathcal{R}_{a}$.
Throughout the report, we use the following notation. $n$ independent
and identically distributed (iid) copies of a random variable and
its realizations are denoted by $X_{0}^{n}$ and $x_{0}^{n}$ respectively.
Length $n$, $\epsilon$-typical set of any random variable $X$,
with distribution $P(X)$ is denoted%
\footnote{The parenthesis is dropped whenever it is obvious%
} by $\mathcal{T}_{\epsilon}^{n}(P(X))$. Throughout the report, for
any set $\mathcal{S}$, we use the shorthand $\{U\}{}_{\mathcal{S}}$
to denote the set $\{U_{i}:i\in\mathcal{S}\}$. Note the difference
between $U_{123}$, which is a single random variable and $\{U\}_{123}$,
which is the set of random variables $\{U_{1},U_{2},U_{3}\}$. In
the following Lemmas, we use the notation $P(A)\doteq2^{-nR}$ to
denote $2^{-n(R+\delta(\epsilon))}\leq P(A)\leq2^{-n(R-\delta(\epsilon))}$
for some $\delta(\epsilon)\rightarrow0$ as $\epsilon\rightarrow0$.
To avoid resolvable but unnecessary complications, we further assume
that there exists at least one joint distribution consistent with
the prescribed per subset distributions for $\mathcal{S}_{1},\mathcal{S}_{2},\ldots,\mathcal{S}_{M}$.

\subsection{Subset Typicality Lemmas}
\begin{lem}
\textbf{Subset Typicality Lemma} :\label{lem:PT}Let $(X_{1},X_{2},\ldots X_{N})$
be $N$ random variables taking values on arbitrary finite alphabets
$(\mathcal{X}_{1},\mathcal{X}_{2},\ldots\mathcal{X}_{N})$ respectively.
Let their marginal distributions be $P_{1}(X_{1}),P_{2}(X_{2})\ldots,P_{N}(X_{N})$
respectively. Let $\mathcal{S}_{1},\mathcal{S}_{2}\ldots\mathcal{S}_{M}$
be $M$ subsets of $\{1,2,\ldots,N\}$ and for all $j\in\{1,2,\ldots,M\}$,
let $P_{\mathcal{S}_{j}}(\{X\}_{\mathcal{S}_{j}})$ be any given joint
distribution for $\{X\}_{\mathcal{S}_{i}}$ consistent with each other
and with the given marginal distributions. Generate sequences $x_{1}^{n},x_{2}^{n}\ldots x_{N}^{n}$,
each independent of the other, where $x_{i}^{n}$ is drawn iid according
to the marginal distribution $P_{i}(X_{i})$, i.e., $x_{i}^{n}\sim\prod_{l=1}^{n}P_{i}(x_{il})$.
Then,\begin{eqnarray}
P\left(\{x\}_{\mathcal{S}_{j}}^{n}\in\mathcal{T}_{\epsilon}^{n}\left(P_{\mathcal{S}_{j}}(\{X\}_{\mathcal{S}_{j}})\right),\forall j\in\{1\ldots M\}\right)\nonumber \\
\doteq2^{-n(\sum_{i=1}^{N}H(X_{i})-H(P^{*}))}\label{eq:PT_TP}\end{eqnarray}
 where $P^{*}$ is a distribution over $(\mathcal{X}_{1},\mathcal{X}_{2}\ldots,\mathcal{X}_{N})$
which satisfies:\begin{equation}
P^{*}=\arg\max_{\tilde{P}}H\left(\tilde{P}\right)\label{eq:PT_Given_1}\end{equation}
 subject to $\tilde{P}(\{X\}_{\mathcal{S}_{j}})=P_{\mathcal{S}_{j}}(\{X\}_{\mathcal{S}_{j}})\,\,\forall j\in\{1\ldots M\}$. 
\end{lem}
This Lemma essentially says that the total number of sequence tuples
$(x_{1}^{n},x_{2}^{n}\ldots x_{N}^{n})$ generated according to their
respective marginals which are jointly $\epsilon-$typical according
to $P_{\mathcal{S}_{i}}(\{X\}_{\mathcal{S}_{i}})\,\,\forall i$ within
subsets $\mathcal{S}_{1},\mathcal{S}_{2}\ldots\mathcal{S}_{M}$, is
approximately $2^{nH(P^{*})}$ where $P^{*}$ is the \textbf{maximum
entropy} distribution subject to the constraint that the joint density
within subset $S_{i}$ is $P_{\mathcal{S}_{j}}(\{X\}_{\mathcal{S}_{j}})\forall j$. 
\begin{proof}
To prove this Lemma, we resort to Sanov's theorem (\cite{Cover-book}
Theorem 11.4.1) from the theory of large deviations. Sanov's theorem
states that for any distribution $Q(X)$ and for any subset of probability
distributions $\mathcal{E}\subseteq\mathcal{P}$, where $\mathcal{P}$
denotes the universe of the PMFs over the alphabets of $X$:\begin{equation}
Q^{n}(\mathcal{E})\doteq2^{-nD(P^{*}||Q)}\label{eq:PT_Pf_1}\end{equation}
for sufficiently large $n$, where $P^{*}$ is the distribution closest
in relative entropy to $Q$ in $\mathcal{E}$ and $Q^{n}(\mathcal{E})$
denotes the probability that an iid sequence generated according to
$Q(X)$ is $\epsilon-$typical with respect to some distribution in
$\mathcal{E}$. We set $Q(\cdot)=\prod_{i=1}^{N}P_{i}(X_{i})$ and
$\mathcal{E}$ as the set of all distributions over $(\mathcal{X}_{1},\mathcal{X}_{2},\ldots\mathcal{X}_{N})$
satisfying the given constraints. Then it follows from Sanov's theorem
that the probability of $(x_{1}^{n}\ldots x_{N}^{n})$ being $\epsilon-$typical
according to some distribution satisfying the given constraints is
approximately $2^{-nD(P^{*}||\prod_{i=1}^{N}P_{i}(X_{i}))}$, where
$P^{*}$ is the distribution having minimum relative entropy to $\prod_{i=1}^{N}P_{i}(X_{i})$
and satisfying the given constraints. However, all such distributions
have the same marginal distributions $P_{i}(X_{1}),P_{i}(X_{2})\ldots,P_{i}(X_{N})$.
Hence minimizing relative entropy is equivalent to maximizing the
joint entropy leading to $P^{*}$ as defined in (\ref{eq:PT_Given_1}).
Therefore we have:\begin{eqnarray*}
P\left(\{x\}_{\mathcal{S}_{i}}^{n}\in\mathcal{T}_{\epsilon}^{n}(\{X\}_{\mathcal{S}_{i}}),\,\forall i\right)=Q^{n}(\mathcal{E})\doteq2^{-nD(P^{*}||Q)}\\
\doteq2^{-n(\sum_{i=1}^{N}H(X_{i})-H(P^{*}))}\end{eqnarray*}
\begin{equation}
\end{equation}
where the last equality follows because $P^{*}$ satisfies the given
marginals. 
\end{proof}
We note that a particular instance of Lemma \ref{lem:PT} was derived
in \cite{Diggavi_ISIT}. However, as it turns out, for the setup they
consider, this Lemma does not help in deriving an improved achievable
region. In the following lemma, we establish the conditional version
of Lemma \ref{lem:PT}. Note that Lemma \ref{cor:PT} is not used
in proving Theorems \ref{thm:main_tm} or \ref{lem:eg}, but will
play a crucial role in the application of these results to more general
multi-terminal source coding scenarios (as we will see in section
\ref{sec:GW}). 
\begin{lem}
\textbf{Conditional Subset Typicality Lemma} :\label{cor:PT}Let random
variables $(X_{1},X_{2},\ldots X_{N})$, sets $\mathcal{S}_{1},\mathcal{S}_{2}\ldots\mathcal{S}_{M}$
and joint densities $P_{\mathcal{S}_{j}}(\{X\}_{\mathcal{S}_{j}})$
be defined as in Lemma \ref{lem:PT}. Let the sequences $(x_{1}^{n}\ldots x_{N}^{n})$
be generated such that each sequence is generated conditioned on a
subset of already generated sequences $\{x\}_{\mathcal{A}_{i}}^{n}$
and independent of the rest, where $(i,\mathcal{A}_{i})\in\mathcal{S}_{j}$
for some $j\in\{1,\ldots,M\}$. Then we have:\begin{eqnarray}
P\left(\{x\}_{\mathcal{S}_{i}}^{n}\in\mathcal{T}_{\epsilon}^{n}(\{X\}_{\mathcal{S}_{i}})\,\,\forall i\in\{1\ldots M\}\right)\doteq\nonumber \\
2^{-n(\sum_{i=1}^{N}H(X_{i}|\{X\}_{\mathcal{A}_{i}})-H(P^{*}))}\label{eq:cond_typ_eq}\end{eqnarray}
 where $P^{*}$ satisfies (\ref{eq:PT_Given_1}). \end{lem}
\begin{proof}
The proof follows in very similar lines to that of Lemma \ref{lem:PT}
by setting $Q(\cdot)=\prod_{i=1}^{N}P(X_{i}|X_{\mathcal{A}_{i}})$,
as conditioning on $x_{\mathcal{A}_{i}}^{n}$ only introduces further
constraints, which are redundant, as $P_{\mathcal{S}_{j}}(\{X\}_{\mathcal{S}_{j}})$
are consistent with each other and $(i,\mathcal{A}_{i})\in\mathcal{S}_{j}$
for some $j\in\{1,\ldots,M\}$. 
\end{proof}

\subsection{Simultaneous Optimality of Functionals}

In this section we will show that simultaneous optimality of all function
$H(\tilde{P}(\{X\}_{\mathcal{J}}))$ is in fact achievable leading
to a new achievable rate region for the problem stated in the introduction. 
\begin{thm}
\label{thm:main_tm} Let random variables $(X_{1},X_{2},\ldots X_{N})$,
sets $\mathcal{S}_{1},\mathcal{S}_{2}\ldots\mathcal{S}_{M}$ and joint
densities $P_{\mathcal{S}_{j}}(\{X\}_{\mathcal{S}_{j}})$ be defined
as in Lemma \ref{lem:PT}. For each $i\in\{1,2\ldots,M\}$, let $x_{i}^{n}(m_{i})$
$m_{i}\in\{1,\ldots,2^{nR_{i}}\}$ be independent sequences drawn
iid according to the respective marginals, i.e., $x_{i}^{n}(m_{i})\sim\prod_{l=1}^{n}P_{i}(x_{il}(m_{i}))$
$\forall m_{i}\in\{1,\ldots,2^{nR_{i}}\}$. Then $\forall\epsilon>0$,
$\exists\delta(\epsilon)$ such that $\delta(\epsilon)\rightarrow0$
as $\epsilon\rightarrow0$ and,\begin{eqnarray}
P\biggl(\{x\}_{\mathcal{S}_{j}}^{n}\left(\{m\}_{\mathcal{S}_{j}}\right)\in\mathcal{T}_{\epsilon}^{n}\left(P_{\mathcal{S}_{j}}(\{X\}_{\mathcal{S}_{j}})\right)\,\,\forall j\nonumber \\
\mbox{ for some }\{m_{1},m_{2}\ldots,m_{N}\}\biggr)\geq1-\delta(\epsilon)\label{eq:main_tm_prob}\end{eqnarray}
if, $(R_{1},R_{2}\ldots,R_{N})$ satisfy the following conditions
$\forall\mathcal{J}\subseteq\{1,2,\ldots,N\}$:\begin{equation}
\sum_{i\in\mathcal{J}}R_{i}\geq\sum_{i\in\mathcal{J}}H(X_{i})-H^{*}(\{X\}_{\mathcal{J}})+\epsilon\label{eq:main_tm_rate_cond}\end{equation}
where, \begin{equation}
H^{*}\left(\{X\}_{\mathcal{J}}\right)=\max_{\tilde{P}(\{X\}_{\mathcal{J}})}H\left(\tilde{P}(\{X\}_{\mathcal{J}})\right)\label{eq:main_tm_max_entropy}\end{equation}
where $\tilde{P}(\{X\}_{\mathcal{J}})$ satisfies:\begin{eqnarray}
\tilde{P}\left(\{X\}_{\mathcal{J}\cap\mathcal{S}_{j}}\right)=P\left(\{X\}_{\mathcal{J}\cap\mathcal{S}_{j}}\right)\forall j\in\{1\ldots M\}\label{eq:main_tm_prp_cond}\end{eqnarray}
We denote the rate region in (\ref{eq:main_tm_rate_cond}) by $\mathcal{R}_{a}^{*}$.\end{thm}
\begin{rem}
Note that $H^{*}(\{X\}_{\mathcal{J}})=H(P(\{X\}_{\mathcal{J}}))$
if $\mathcal{J}\subseteq\mathcal{S}_{j}$ for some $\mathcal{S}_{j}$.
Hence for all $\mathcal{J}$ such that $\mathcal{J}\subseteq\mathcal{S}_{j}$
for some $j$, the corresponding inequalities in Theorem \ref{thm:main_tm}
and equations (\ref{eq:MCL2}) are the same. However this theorem
asserts that for every other $\mathcal{J}$, the functionals in (\ref{eq:MCL2})
can be replaced with the `maximum joint entropy' subject to the given
subset distributions which involve only the random variables $\{X\}_{\mathcal{J}}$.
It is very important to note that the maximum entropy distributions
for two different subsets $X_{\mathcal{J}_{1}}$ and $X_{\mathcal{J}_{2}}$,
$\mathcal{J}_{1},\mathcal{J}_{2}\subseteq\{1,2,\ldots,N\}$, may not
even correspond to any valid joint distribution over $(X_{1},X_{2},\ldots,X_{N})$.
This is precisely what provides the additional leeway in achieving
points which are strictly outside (\ref{eq:MCL2}) as illustrated
in Theorem \ref{lem:eg}. A pictorial representation of the above
theorem is shown in Fig. \ref{fig:Pict_Sanov}.\end{rem}
\begin{proof}
We are interested in finding conditions on rates so that the probability
in (\ref{eq:main_tm_prob}) approaches $1$. Denote the event $\mathcal{E}=\nexists\{m_{1},m_{2}\ldots,m_{N}\}:\{x\}_{\mathcal{S}_{j}}^{n}\left(\{m\}_{\mathcal{S}_{j}}\right)\in\mathcal{T}_{\epsilon}^{n}\left(P_{\mathcal{S}_{j}}(\{X\}_{\mathcal{S}_{j}})\right)\,\,\forall j$.
We want to make $P(\mathcal{E})\rightarrow0$. Let $\mathcal{N}$
denote the set $\{1,2,\ldots,N\}$ and let $(m_{1},m_{2},\ldots,m_{N})=\{m\}_{\mathcal{N}}$
be an index tuple, one from each codebook, such that $m_{i}\in\{1,\ldots,2^{nR_{i}}\}$.
Let $\mathcal{E}(\{m\}_{\mathcal{N}})$ denote the event that $\{x\}_{\mathcal{S}_{j}}^{n}\left(\{m\}_{\mathcal{S}_{j}}\right)\in\mathcal{T}_{\epsilon}^{n}\left(P_{\mathcal{S}_{j}}(\{X\}_{\mathcal{S}_{j}})\right)$
$\forall j$. Define random variables $\chi(\{m\}_{\mathcal{N}})$
such that:\begin{equation}
\chi(\{m\}_{\mathcal{N}})=\begin{cases}
1 & \mbox{if}\,\,\mathcal{E}(\{m\}_{\mathcal{N}})\,\,\mbox{occurs}\\
0 & \mbox{else}\end{cases}\label{eq:P_1}\end{equation}
and random variable $\chi=\sum_{\{m\}_{\mathcal{N}}}\chi(\{m\}_{\mathcal{N}})$.
Then we have $P(\mathcal{E})=P(\chi=0)$. From Chebyshev's inequality,
it follows that:\begin{eqnarray}
P(\mathcal{E})=P(\chi=0)\leq P\left[|\chi-E(\chi)|\geq E(\chi)/2\right]\label{eq:P_3}\\
\leq\frac{4\mbox{Var}(\chi)}{\left(E(\chi)\right)^{2}}=\frac{4\left(E(\chi^{2})-\left(E(\chi)\right)^{2}\right)}{\left(E(\chi)\right)^{2}}\nonumber \end{eqnarray}
We next bound $E(\chi)$ and $E(\chi^{2})$ using Lemma \ref{lem:PT}.
First we write $E(\chi)$ as:\begin{eqnarray}
E(\chi) & = & 2^{n\sum_{i=1}^{N}R_{i}}P(\mathcal{E}(\{m\}_{\mathcal{N}}))\label{eq:P_4}\end{eqnarray}
for any $\{m\}_{\mathcal{N}}$ because all the sequences are drawn
independent of each other. Next towards bounding $E(\chi^{2})$, note
that:\begin{equation}
E(\chi^{2})=\sum_{\{m\}_{\mathcal{N}}}\sum_{\{l\}_{\mathcal{N}}}P\left(\mathcal{E}(\{m\}_{\mathcal{N}}),\mathcal{E}(\{l\}_{\mathcal{N}})\right)\label{eq:P_5}\end{equation}
Let $\{m\}_{\mathcal{Q}}=\{l\}_{\mathcal{Q}}$ and $\{m\}_{\mathcal{N}-\mathcal{Q}}\neq\{l\}_{\mathcal{N}-\mathcal{Q}}$
for some $\mathcal{Q}\subseteq\mathcal{N},\mathcal{Q}\neq\phi$ where
$\phi$ denotes a null-set. Then,\begin{eqnarray}
P\left(\mathcal{E}(\{m\}_{\mathcal{N}}),\mathcal{E}(\{l\}_{\mathcal{N}})\right)=\Biggl\{ P\left(\mathcal{E}(\{m\}_{\mathcal{Q}})\right)\nonumber \\
P\left(\mathcal{E}(\{m\}_{\mathcal{N}})\Bigl|\mathcal{E}(\{m\}_{\mathcal{Q}})\right)^{2}\Biggr\}\label{eq:P_6-1}\end{eqnarray}
where $\mathcal{E}(\{m\}_{\mathcal{Q}})$ denotes the event that $\{x\}_{\mathcal{S}_{j}\cap\mathcal{Q}}^{n}\left(\{m\}_{\mathcal{S}_{j}\cap\mathcal{Q}}\right)\in\mathcal{T}_{\epsilon}^{n}\left(P_{\mathcal{S}_{j}\cap\mathcal{Q}}(\{X\}_{\mathcal{S}_{j}\cap\mathcal{Q}})\right)$
$\forall j$, as conditional on $\{x\}_{\mathcal{Q}}^{n}(\{m\}_{\mathcal{Q}})$,
sequences $\{x\}_{\mathcal{N}-\mathcal{Q}}^{n}(\{m\}_{\mathcal{N}-\mathcal{Q}})$
and $\{x\}_{\mathcal{N}-\mathcal{Q}}^{n}$ $(\{l\}_{\mathcal{N}-\mathcal{Q}})$
are drawn independently from the same distribution. The above expression
can be rewritten as:\begin{eqnarray}
P\left(\mathcal{E}(\{m\}_{\mathcal{N}}),\mathcal{E}(\{l\}_{\mathcal{N}})\right)=\Biggl\{ P\left(\mathcal{E}(\{m\}_{\mathcal{Q}})\right)\nonumber \\
\times\left(\frac{P\left(\mathcal{E}(\{m\}_{\mathcal{N}})\right)}{P\left(\mathcal{E}(\{m\}_{\mathcal{Q}})\right)}\right)^{2}\Biggr\}\label{eq:P_6}\end{eqnarray}
If $\mathcal{Q}=\phi$, we have $P\left(\mathcal{E}(\{m\}_{\mathcal{N}}),\mathcal{E}(\{l\}_{\mathcal{N}})\right)=\left(P\left(\mathcal{E}(\{m\}_{\mathcal{N}})\right)\right)^{2}$.
Hence, we can write $Var(\chi)$ as:\begin{eqnarray}
Var(\chi)=\sum_{\mathcal{Q}\subseteq\mathcal{N},\mathcal{Q}\neq\phi}\Biggl\{2^{n\sum_{i\in\mathcal{Q}}R_{i}+2n\sum_{i\in\mathcal{N}-\mathcal{Q}}R_{i}}\nonumber \\
\times P\left(\mathcal{E}(\{m\}_{\mathcal{Q}})\right)\left(\frac{P\left(\mathcal{E}(\{m\}_{\mathcal{N}})\right)}{P\left(\mathcal{E}(\{m\}_{\mathcal{Q}})\right)}\right)^{2}\Biggr\}\label{eq:P_7}\end{eqnarray}
Note that the $\mathcal{Q}=\phi$ term gets cancelled with the $ $`$\left(E(\chi)\right)^{2}$'
terms in $Var(\chi)$ (see \cite{VKG} for a similar argument).

On substituting (\ref{eq:P_4}) and (\ref{eq:P_7}) in (\ref{eq:P_3}),
and noting that for any $\mathcal{Q}\subseteq\mathcal{N}$, $\mathcal{Q}\neq\phi$,
we can write $P(\mathcal{E}(\{m\}_{\mathcal{N}}))=P\left(\mathcal{E}(\{m\}_{\mathcal{Q}})\right)\frac{P\left(\mathcal{E}(\{m\}_{\mathcal{N}})\right)}{P\left(\mathcal{E}(\{m\}_{\mathcal{Q}})\right)}$,
we have:\begin{equation}
P(\mathcal{E})\leq4\sum_{\mathcal{Q}\subseteq\mathcal{N},\mathcal{Q}\neq\phi}2^{-n\sum_{i\in\mathcal{Q}}R_{i}}\left(P(\mathcal{E}(\{m\}_{\mathcal{Q}}))\right)^{-1}\label{eq:P8}\end{equation}
Next, invoking Lemma \ref{lem:PT}, we bound $P(\mathcal{E}(\{m\}_{\mathcal{Q}}))$
as:\begin{equation}
P(\mathcal{E}(\{m\}_{\mathcal{Q}}))\geq2^{-n\left(\sum_{i\in\mathcal{Q}}H(X_{i})-H^{*}(\{X\}_{\mathcal{Q}}))\right)-n\delta(\epsilon)}\label{eq:P9}\end{equation}
On substituting (\ref{eq:P9}) in (\ref{eq:P8}), it follows that
$P(\mathcal{E})\rightarrow0$ as $n\rightarrow\infty$ if $R_{i}$
satisfy (\ref{eq:main_tm_rate_cond}).
\end{proof}
\begin{figure}
\centering\includegraphics[scale=0.25]{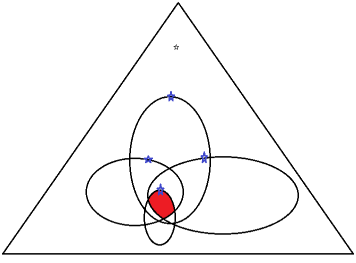}\caption{\label{fig:Pict_Sanov}Pictorial representation of Theorem \ref{thm:main_tm}:
The triangle denotes the simplex of all joint distributions over $(X_{1},X_{2},\ldots,X_{N})$.
The black star denotes the joint distribution representing the product
of marginals (codebook generation). Each loop represent the set of
all joint distributions satisfying the conditions imposed on $\{X\}_{\mathcal{J}}$
for some $\mathcal{J}$. The intersection of all the loops (red region)
represents the set of joint distributions satisfying all the conditions.
The blue stars represent the joint distributions which maximize functionals
$H(P(\{X\}_{\mathcal{J}}))$ (equivalently, minimize the relative
entropy with the product of marginals as seen from Sanov's theorem)
subject to the conditions on $\{X\}_{\mathcal{J}}$. Theorem \ref{thm:main_tm}
asserts that a separate joint distribution for each $\mathcal{J}$
can be chosen from the corresponding loop (blue stars) and hence all
the functionals $H(P(\{X\}_{\mathcal{J}}))$ can be set to their respective
maxima simultaneously.}

\end{figure}

\subsection{Strict Improvement}
\begin{thm}
\label{lem:eg}(i) The region in Theorem \ref{thm:main_tm} subsumes
the region in (\ref{eq:MCL3}). i.e,\begin{equation}
\mathcal{R}_{a}\subseteq\mathcal{R}_{a}^{*}\label{eq:strict_l_1}\end{equation}
(ii) There exist scenarios for which the region in Theorem \ref{thm:main_tm}
can be strictly larger than the region in (\ref{eq:MCL3}). i.e.,\begin{equation}
\mathcal{R}_{a}^{*}\supset\mathcal{R}_{a}\label{eq:strick_l_2}\end{equation}
\end{thm}
\begin{proof}
The first half of the Theorem follows directly because $H^{*}(\{X\}_{\mathcal{J}})\geq H(\{X\}_{\mathcal{J}})\,\,\forall\mathcal{J}$
for any joint distribution satisfying the given distributions within
subsets. To prove (ii) we provide an example for which $\mathcal{R}_{a}^{*}$
has points which are not part of $\mathcal{R}_{a}$. Consider the
following example of $4$ binary random variables $(X_{1},X_{2},X_{3},X_{4})$.
$X_{1},X_{2}$ and $X_{3}$ are distributed $bern(\frac{1}{2})$ and
$X_{4}$ is distributed $bern(\frac{3}{4})$, where $bern(p)$ denotes
a Bernoulli random variable with $P(0)=p$ and $P(1)=1-p$. Let $\mathcal{S}_{1},\mathcal{S}_{2}\ldots,\mathcal{S}_{6}$
be all possible subsets of $\{1,2,3,4\}$ of cardinality $2$. Let
$P_{\mathcal{S}_{j}}(\{X\}_{\mathcal{S}_{j}})$ be such that $(X_{1},X_{2},X_{3})$
are pairwise independent and the pairwise PMF of $(X_{i},X_{4})$
$\forall i\in\{1,2,3\}$ is given in Table \ref{tab:Pairwise}. Note
that these pairwise densities are satisfied by at lease one joint
density obtained by the following operations : $X_{3}=X_{1}\oplus X_{2}$
and $X_{4}=X_{1}\bullet X_{2}$, where $X_{1}$ and $X_{2}$ are independent
$bern(\frac{1}{2})$ random variables and `$\oplus$' and `$\bullet$'
denote `bit-exor' and `bit-and' operations respectively. 

\begin{table}
\centering\begin{tabular}{|c|c|c|c|c|}
\hline 
$(x_{i},x_{4})$ & $0,0$ & $0,1$ & $1,0$ & $1,1$\tabularnewline
\hline
\hline 
$P(x_{i},x_{4})$ & $\unitfrac{1}{2}$ & $0$ & $\unitfrac{1}{4}$ & $\unitfrac{1}{4}$\tabularnewline
\hline
\end{tabular}\caption{\label{tab:Pairwise}Pairwise PMF of $(X_{i},X_{4})\,\,\forall i\in\{1,2,3\}$}

\end{table}

Observe that maximizing the entropy over $(X_{1},X_{2},X_{3})$ subject
to their respective pairwise densities makes them mutually independent.
However, there exists \textit{no} joint distribution over $(X_{1},X_{2},X_{3},X_{4})$
satisfying all the pairwise conditions which makes $(X_{1},X_{2},X_{3})$
mutually independent. This intuition is in fact sufficient to see
that $\mathcal{R}_{a}^{*}\supset\mathcal{R}_{a}$. However to be more
rigorous, we first rewrite the achievable region $\mathcal{R}_{a}^{*}$
for this example as:\begin{eqnarray}
R_{i}+R_{4} & \geq & H_{b}(\frac{1}{4})-\frac{1}{2}H_{b}(\frac{1}{2})\nonumber \\
R_{i}+R_{j}+R_{4} & \geq & 2+H_{b}(\frac{1}{4})-H^{*}(X_{i},X_{j},X_{4})\nonumber \\
\sum_{i=1}^{4}R_{i} & \geq & 3+H_{b}(\frac{1}{4})-H^{*}(\{X\}_{1,2,3,4})\label{eq:eq_Rs-1}\end{eqnarray}
$\forall i,j\in\{1,2,3\}$ where $H_{b}(\cdot)$ denotes the binary
entropy function and $\{X\}_{1,2,3,4}=\{X_{1},X_{2},X_{3},X_{4}\}$. 

We consider the following corner point of (\ref{eq:eq_Rs-1}), $A=(0,0,0,3+H_{b}(\frac{1}{4})-H^{*}(\{X\}_{1,2,3,4}))$.
It is sufficient for us to prove that $A$$\notin\mathcal{R}_{a}$
. Note that, if $R_{1}=R_{2}=R_{3}=0$, $(X_{1},X_{2},X_{3})$ must
be mutually independent (which in-turn satisfies the pairwise independence
conditions). To prove that $A\notin\mathcal{R}_{a}$ , we will show
that there cannot exist \textit{any} joint PMF over $(X_{1},X_{2},X_{3},X_{4})$
satisfying all pairwise distributions and for which $(X_{1},X_{2},X_{3})$
are mutually independent. Let us suppose that such a joint PMF exists.
Denote the conditional PMF $P(X_{4}=0|x_{1},x_{2},x_{3})=\alpha_{x_{1}x_{2}x_{3}}$,
$x_{1},x_{2},x_{3}\in\{0,1\}$. As $(X_{1},X_{2},X_{3})$ are assumed
to be mutually independent, the joint distribution $P_{X_{1},X_{2},X_{3},X_{4}}(x_{1},x_{2},x_{3},1)=\frac{1-\alpha_{x_{1}x_{2}x_{3}}}{8}$.
The pairwise distribution of $(X_{1},X_{4})$ (from Table \ref{tab:Pairwise})
is such that $P_{X_{i},X_{4}}(0,1)=0$ $\forall i\in\{1,2,3\}$. This
leads to the conclusion that $\alpha_{x_{1}x_{2}x_{3}}=1$ if any
one of $x_{1},x_{2},x_{3}$ is $0$. We are only left with finding
$\alpha_{111}$. Further, we want $P_{X_{1},X_{4}}(1,1)=\frac{1}{4}$,
i.e. $\sum_{x_{2},x_{3}}P_{X_{1},X_{2},X_{3},X_{4}}(1,x_{2},x_{3},1)=\sum_{x_{2},x_{3}}\frac{1-\alpha_{1x_{2}x_{3}}}{8}=\frac{1}{4}$.
One substituting, we have $\alpha_{111}=2$. As $\alpha_{x_{1},x_{2},x_{3}}$s
are conditional probabilities, this leads to a contradiction and proves
that there cannot exist a joint distribution with $(X_{1},X_{2},X_{3})$
being mutually independent. Therefore $\mathcal{R}_{a}^{*}\supset\mathcal{R}_{a}$,
proving the second half of the Theorem. 
\end{proof}

\section{Application to Multi-User Gray-Wyner Network\label{sec:GW}}

\begin{figure}
\centering\includegraphics[scale=0.35]{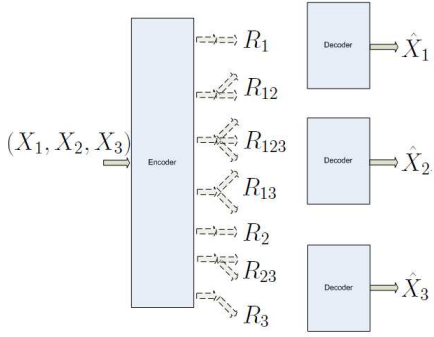}\caption{3-user Gray-Wyner network: There is a unique branch from the encoder
to every subset of the decoders\label{fig:3GW}}

\end{figure}

We finally apply the results in Theorem \ref{thm:main_tm} to obtain
a new achievable region for the multi-user Gray-Wyner network. To
illustrate the applicability and to maintain simplicity in notation,
we only consider the 3-user lossless Gray-Wyner network here. However
the approach can be extended directly to the general $L-$user setting
and to incorporate distortions. Note that the formal definition of
an achievable rate region closely resembles that in \cite{GW}, with
obvious generalization to the 3 user setting as shown in Fig. \ref{fig:3GW}.
We omit the details here due to space constraints. We further note
that the rate region is in general 7 dimensional, with the following
rates: $(R_{1},R_{2},R_{3},R_{12},R_{13},R_{23},R_{123})$.
\begin{cor}
Let $(X_{1},X_{2},X_{3})$ be the random variables with joint distribution
$P(X_{1},X_{2},X_{3})$ observed by the encoder. Let $(U_{123},U_{12},U_{13},U_{23})$
be random variables jointly distributed with $(X_{1},X_{2},X_{3})$
with conditional distribution $P(U_{123},U_{12},U_{13},U_{23}|X_{1},X_{2},X_{3})$
and taking values over arbitrary finite alphabets. Define subsets
$\mathcal{S}_{1}=\{U_{123},U_{12},U_{13}\}$, $\mathcal{S}_{2}=\{U_{123},U_{12},U_{23}\}$,
$\mathcal{S}_{3}=\{U_{123},U_{13},U_{23}\}$. The rate region for
the 3-user lossless Gray-Wyner network contains all the rates such
that $\forall(i,j,k)\in\{1,2,3\}$ and $i<j$, $i<k$, \begin{eqnarray}
R_{123} & \geq & H(U_{123})-H^{*}(U_{123}|\mathbf{X})\nonumber \\
R_{123}+R_{ij} & \geq & H(U_{123},U_{ij})\nonumber \\
 &  & -H^{*}(U_{123},U_{ij}|\mathbf{X})\nonumber \\
R_{123}+R_{ij}+R_{ik} & \geq & H(U_{123})-H^{*}(\{U\}_{123,ij,ik}|\mathbf{X})\nonumber \\
 &  & +H(U_{ij}|U_{123})+H(U_{ik}|U_{123})\nonumber \\
R_{123}+\sum_{i<j}R_{ij} & \geq & H(U_{123})+\sum_{i<j}H(U_{ij}|U_{123})\nonumber \\
 &  & -H^{*}(U_{123},U_{12},U_{23},U_{13}|\mathbf{X})\nonumber \\
R_{i} & \geq & H(X_{i}|\{U\}_{\mathcal{J}:i\in\mathcal{J}})\label{eq:3GW_rate_cond}\end{eqnarray}
where $\mathbf{X}=\{X_{1},X_{2},X_{3}\}$ and $H^{*}(\{U\}_{\mathcal{J}}|\mathbf{X})$
is given by:\begin{equation}
\max_{\tilde{P}(\{U\}_{\mathcal{J}},\{X\}_{1,2,3})}H\left(\tilde{P}\left(\{U\}_{\mathcal{J}}\right)\Bigl|\mathbf{X}\right)\end{equation}
where \textup{$\tilde{P}(\{U\}_{\mathcal{J}}\bigl|\mathbf{X})$} satisfies:\begin{eqnarray}
\tilde{P}\left(\{U\}_{\mathcal{J}\cap\mathcal{S}_{j}},X_{j}\right) & = & P\left(\{U\}_{\mathcal{J}\cap\mathcal{S}_{j}},X_{j}\right)\,\,\forall j\end{eqnarray}
The closure of the achievable rates over all conditional distributions
$P(U_{123},U_{12},$ $U_{13},U_{23}|X_{1},X_{2},X_{3})$ is an achievable
region for the 3-user lossless Gray-Wyner network.\end{cor}
\begin{proof}
A codebook for $U_{123}$ consisting of $2^{nR_{123}}$ codewords
is generated according to the marginal $P(U_{123})$. Conditioned
on each codeword of $U_{123}$, independent codebooks are generated
for $U_{12},U_{13}$ and $U_{23}$ at rates of $R_{12},R_{13}$ and
$R_{23}$ according to their respective conditional distributions
$P(U_{12}|U_{123})$, $P(U_{13}|U_{123})$ and $P(U_{23}|U_{123})$.
If the rates satisfy (\ref{eq:3GW_rate_cond}), then there always
exists a codeword tuple, one from each codebook, denoted by $(u_{123}^{n},u_{12}^{n},u_{13}^{n},u_{23}^{n})$,
such that the following subsets of sequences are jointly typical according
to their respective subset joint densities: $(x_{1}^{n},u_{123}^{n},u_{12}^{n},u_{13}^{n})$,
$(x_{2}^{n},u_{123}^{n},u_{12}^{n},u_{23}^{n})$ and $(x_{3}^{n},u_{123}^{n},u_{13}^{n},u_{23}^{n})$.
The proof follows rather directly from Lemmas \ref{lem:PT}, \ref{cor:PT}
and Theorem \ref{thm:main_tm} as $U_{123}$ is part of $\mathcal{S}_{1},\mathcal{S}_{2}$
and $\mathcal{S}_{3}$. The last constraint in (\ref{eq:3GW_rate_cond})
denotes the minimum rate of the bin indices required to achieve lossless
reconstruction at each sink given that all the codewords received
at any sink are jointly typical. 
\end{proof}

\section{Discussion}

We note that the conditions in (\ref{eq:3GW_rate_cond}) ensure joint
typicality of source sequence $X_{i}^{n}$ only with the codewords
which reach sink $i$. However an alternate achievable region (which
is subsumed in the above region) can be derived using results of the
general $L-$channel MD problem in \cite{VKG} which extends the principles
underlying (\ref{eq:MCL3}) to the multiple descriptions framework.
Due to the inherent structure of the MD problem, joint typicality
of all the transmitted codewords is necessary. However imposing such
a constraint limits the performance of systems that do not explicitly
require such conditions. Note that, although we have not proved formally
that the new region for the multi-user Gray-Wyner network is strictly
larger than that derivable from the results in \cite{VKG}, Theorem
\ref{lem:eg} suggests that for general sources, there exist points
which are strictly outside. It is important to note that implications
of the results we derived may not always lead to a strictly larger
achievable region. A classic example of this setting is the 2 user
Gray-Wyner network \cite{GW} for which the complete rate-distortion
region can be achieved even if joint typicality of all the codewords
is imposed. This is because, in the 2-user scenario, there is no inherent
conflict between maximum entropy distributions of different subsets
of random variables. However, in the $L-$user setting (as seen in
Theorem \ref{lem:eg}), such a conflict arises and maintaining joint
typicality only within subsets plays a paramount role in deriving
improved achievable regions.

\end{document}